\documentclass[10pt,aps,pre,twocolumn,superscriptaddress,nofootinbib,floatfix]{revtex4-2}
\usepackage{amsmath,amssymb,amsthm}
\usepackage{graphicx}
\usepackage{booktabs}
\usepackage[colorlinks=true,linkcolor=blue,citecolor=blue,urlcolor=blue]{hyperref}

\newtheorem{theorem}{Theorem}
\newtheorem{proposition}{Proposition}
\newtheorem{conjecture}{Conjecture}
\newcommand{\R}{\mathbb{R}}
\newcommand{\Q}{\mathbb{Q}}
\newcommand{\ev}[1]{\langle #1 \rangle}
\newcommand{\pcert}{p_{\mathrm{c}}}
\newcommand{\COP}{\mathcal{C}}
\newcommand{\PSDN}{\mathcal{P}\!+\!\mathcal{N}}
\newcommand{\Kc}{\mathcal{K}}

\begin{document}

\title{Stationarity is not enough: tightness of the quantum mechanical
bootstrap and the copositive cone}

\author{Daniel Keren}
\email{dkeren@cs.haifa.ac.il}
\affiliation{Department of Computer Science, University of Haifa,
Haifa, Israel}

\date{\today}

\begin{abstract}
The numerical bootstrap for quantum mechanics tests positivity of a candidate state only against sums of squares (SOS), whereas every physical state assigns nonnegative expectation to every pointwise nonnegative polynomial. In one dimension the two coincide; in two or
more they do not.  Whether this gap is realized depends sharply on which
constraint set is imposed.  For the \emph{stationary} bootstrap, which
imposes $\ev{[H,\mathcal{O}]}=0$ and is the relaxation appropriate to thermal and
mixed states, we exhibit a two-dimensional quartic double well potential and a
moment vector that satisfies every level-three stationary constraint exactly, has a
positive definite moment matrix, and yet assigns a negative expectation to a
polynomial nonnegative on $\R^2$; it is therefore the moment sequence of
no state.  All data are rational and every step is verified in exact arithmetic.
For the \emph{eigenstate} bootstrap, which additionally imposes
$\ev{\mathcal{O}H}=E\ev{\mathcal{O}}$, the same search finds no violation in
any of five settings spanning two, three and five degrees of freedom and
truncation levels three and four, tested in each case against a family of
separating polynomials that is complete at the relevant degree or matched to
the copositive obstruction.  The single exception occurs at a
truncation so low that only one eigenstate constraint survives, and the theory
here accounts for it.  We identify the mechanism: the eigenstate constraints
bound the high momentum moments, otherwise unbounded on the feasible set, and
it is those directions that reach the region between the SOS and nonnegative cones.  Finally, under the reflection symmetries of a
typical potential the relevant obstruction is copositivity rather than
nonnegativity, which, for the quartic witnesses available at the lowest
truncation, places the first possible failure at five degrees of freedom.  We conjecture that the eigenstate
constraints imply an Archimedean-type bound on the momentum moments, and
formulate the corresponding tightness statement.
\end{abstract}

\maketitle

\section{Introduction}

The bootstrap approach to quantum mechanics~\cite{HHK2020} determines spectra
without solving a differential equation.  For a state $\rho$ and Hamiltonian
$H$, the linear functional $L(A)=\mathrm{tr}(\rho A)$ satisfies
\begin{equation}
  L(A^\dagger A)\ \ge\ 0
  \label{eq:pos}
\end{equation}
for every operator $A$, and, if $\rho$ is stationary,
\begin{equation}
  L([H,A])\ =\ 0 .
  \label{eq:stat}
\end{equation}
If moreover $\rho=|\psi\rangle\langle\psi|$ with $H\psi=E\psi$, then
additionally
\begin{equation}
  L(\mathcal{O}H)\ =\ E\,L(\mathcal{O}) .
  \label{eq:eig}
\end{equation}
Discarding the state and retaining \eqref{eq:pos}--\eqref{eq:eig} for $A$ in a
finite-dimensional span of polynomials in position and momentum yields a
semidefinite program (SDP) in the moments $\ev{q^\alpha p^\beta}$
\cite{BH2021,BH2022a,BH2023}.

The distinction between \eqref{eq:stat} alone and \eqref{eq:stat} together
with \eqref{eq:eig} is standard and is drawn explicitly in the
literature~\cite{BH2023}: the former characterizes a stationary state, the
latter an energy eigenstate.  A Gibbs state $e^{-\beta H}/Z$ is stationary
and is not an eigenstate; for a mixture of two or more energy levels
\eqref{eq:eig} admits no solution $E$ at all.  Thermal bootstraps therefore
impose \eqref{eq:stat} with positivity and thermal inequalities, and no
eigenstate condition~\cite{CGSY2025}.  We will show that the two relaxations
behave differently with respect to tightness.

The construction is a truncated moment problem, and it inherits that problem's
central subtlety.  Requiring $L(A^\dagger A)\ge 0$ for every $A$ in the aforementioned polynomial span
is equivalent to requiring $L$ to be nonnegative on the sums of squares (SOS) formed from
it.  A genuine state
satisfies more: $L(f)\ge0$ for every $f$ pointwise nonnegative.  Hilbert
\cite{Hilbert1888} showed the two classes differ, explicit nonnegative non-SOS
forms are known~\cite{Motzkin,ChoiLam,Robinson}, and asymptotically the 
gap between them increases~\cite{Blekherman}.  In one variable there is
no gap, which is why one-dimensional bootstraps behave so well; in two or more
variables there is.

Whether that gap is ever \emph{realized} by a Schr\"odinger operator has been
open.  Spurious allowed regions have been reported~\cite{spurious},
rigorously establishing convergence has been stated as an open
problem~\cite{central2025}, and a distinct failure -- an operator ambiguity for
potentials mixing polynomial and exponential terms -- has recently been
identified in one dimension~\cite{Morita2026}.  That the nonnegative and
sum-of-squares cones differ in more than one variable is known in the bootstrap
literature and handled there with a multiplier~\cite{NancarrowXin2023}; what
has been missing is an instance in which the difference is attained.  The question is also outside the reach of the
convergence theory for noncommutative polynomial
optimization~\cite{NPO2026,NPOhierarchy}, which requires an Archimedean
condition and hence bounded operators, whereas $[q,p]=i$ admits no
finite-dimensional or bounded representation.

We treat the two relaxations separately.  For the stationary bootstrap we
settle the question at level three: the gap is realized, and we exhibit a
point that realizes it, certified in exact arithmetic.  For the eigenstate
bootstrap we do not settle it, and instead present systematic numerical evidence that no
violation exists, together with a mechanism that accounts for the difference.

Nothing here bears on the extensive body of one-dimensional results obtained
with the method~\cite{BH2021,BH2022a,BH2023,BHreflection,spurious}: for a single
degree of freedom every nonnegative polynomial is a sum of squares, so the
relaxation tests exactly the correct condition and the gap we study is empty.
Our phenomenon requires at least two degrees of freedom.  It is also distinct
from the boundary anomalies of~\cite{BHhalfline}, which arise from operator
domains rather than from the positivity cone.

\section{Setup}
\label{sec:setup}

We consider
\begin{equation}
  H \;=\; \tfrac12\sum_{j=1}^{d} p_j^2 \;+\; V(q_1,\dots,q_d),
  \label{eq:H}
\end{equation}
with $V$ a polynomial for which $H$ is bounded below with a discrete spectrum,
and the level-$K$ bootstrap: the unknowns are the moments
$m_{\alpha\beta}=\ev{q^\alpha p^\beta}$ of degree at most $2K$, subject to
$\ev{1}=1$, hermiticity, positive semidefiniteness of the moment matrix $M_K$
indexed by monomials of degree at most $K$, and
\begin{align}
  \text{(S)}\quad & \ev{[H,\mathcal{O}]}=0, && \deg\mathcal{O}\le 2K-\deg V+2, \notag\\
  \text{(E)}\quad & \ev{\mathcal{O}H}=E\ev{\mathcal{O}}, && \deg\mathcal{O}\le 2K-\deg V .
  \label{eq:constraints}
\end{align}
Fixing $E$ renders every constraint linear and turns feasibility into an SDP;
this is the method used in~\cite{BH2023}.  The energy normalization $\ev{H}=E$ --
the case $\mathcal{O}=1$ of (E) -- is imposed throughout, as it allows 
the energy to enter; we write $\mathcal{F}_K^{\mathrm S}(E)$ for the
feasible set obtained by imposing that together with (S) alone, and
$\mathcal{F}_K^{\mathrm{SE}}(E)$ for the set obtained by imposing (S) and the
entire family (E).  Clearly
$\mathcal{F}_K^{\mathrm{SE}}(E)\subseteq\mathcal{F}_K^{\mathrm S}(E)$.

Each bound in \eqref{eq:constraints} is the largest degree of $\mathcal{O}$
for which the resulting relation involves only moments of degree at most
$2K$, the unknowns of the level-$K$ bootstrap.  The two differ because
$\mathcal{O}H$ has degree $\deg\mathcal{O}+\deg V$, while the commutator
$[H,\mathcal{O}]$ loses two.  For a quartic $V$ at $K=3$, this yields 
$\deg\mathcal{O}\le4$ in (S) and $\deg\mathcal{O}\le2$ in (E).  
A relation whose bound is exceeded involves moments that are not among the
unknowns, and it must be discarded.

\subsection{Two dimensions require momentum moments}
\label{sec:momentum}
Writing $q_1=x$, $q_2=y$, the relation $\ev{[H,f\,p_x]}=0$ gives, for any
polynomial $f(x,y)$,
\begin{align}
  \ev{(\partial_xf)\,p_x^2} + \ev{(\partial_yf)\,p_yp_x}
  \;=\; \ev{f\,\partial_xV} - \tfrac14\ev{\partial_x\nabla^2f},
  \label{eq:hypervirial}
\end{align}
where the expectations are of the real parts.  We have verified
\eqref{eq:hypervirial} in exact rational arithmetic at the moment vector of
Theorem~1, for every monomial $f$ of degree at most three -- the most the
budget allows, since $\mathcal{O}=f\,p_x$ costs a degree -- both sides agree
identically, in real and imaginary parts alike.  That vector satisfies (S) and violates (E), so the check confirms
\eqref{eq:hypervirial} as a consequence of (S) alone: had the derivation
appealed to \eqref{eq:eig} at any step, the identity would fail there.  A
check against a genuine eigenstate could not distinguish the two, since
\eqref{eq:eig} holds there as well.

In one dimension the momentum moment can be eliminated.  The cross term
$\ev{(\partial_yf)\,p_yp_x}$ is absent, and taking $f=x^n$, $p^2=2(H-V)$ together with the
eigenstate condition \eqref{eq:eig} converts $\ev{x^{n-1}p^2}$ into
$2E\ev{x^{n-1}}-2\ev{x^{n-1}V}$.  What remains is a recursion in the single
sequence $\ev{x^k}$ with $E$ as a parameter, which is how one-dimensional
bootstraps are run~\cite{HHK2020}.  Note that this elimination uses
\eqref{eq:eig}, not \eqref{eq:stat} alone.

In two dimensions no such elimination exists.  The left-hand side of
\eqref{eq:hypervirial} carries $\ev{(\partial_xf)p_x^2}$, and $p_x^2$ by itself
is not $2(H-V)$ -- only $p_x^2+p_y^2$ is -- so nothing converts it into
position moments while leaving the cross term untouched.  The multidimensional
bootstrap is therefore a genuinely noncommutative moment problem, in which the
momentum moments are unknowns in their own right; Section~\ref{sec:mechanism}
shows that this is precisely where the relaxation is weakest.

Nevertheless the position moments $\ev{q^a}$ are carried by a principal
submatrix of $M_K$ and are exactly a truncated commutative moment sequence in
$d$ variables.  A polynomial in the positions alone with negative expectation
therefore rules out the entire moment vector, which makes the classical
nonnegative-vs-SOS gap available even though the problem is noncommutative.
This is what we focus on.

We call a polynomial $f$ a \emph{separating witness} for a moment vector $m$
if $f\ge0$ pointwise while $L_m(f)<0$.  Such an $f$ certifies that $m$ is not
the moment sequence of any positive measure, since every measure gives
$L(f)\ge0$; it is a separating hyperplane between $m$ and the moment cone.

The noncommutative structure does not alleviate the problem, although in the
free case it would.  Helton's theorem~\cite{Helton2002} states that a
noncommutative polynomial positive on all tuples of symmetric matrices of all
finite sizes is a sum of squares, so in the free algebra there is no gap at
all; this underlies the noncommutative optimization
programme~\cite{BKP2016}.  However the Weyl algebra is not free: it carries
$[q_j,p_j]=i$, and taking traces shows it has no finite-dimensional
representation, so Helton's hypothesis is never met.

It is worth drawing the consequence sharply, because this delimits where the
phenomenon of this paper can occur at all.  In an algebra with
finite-dimensional representations -- the Clifford algebra of fermionic
modes, or the Pauli algebra of spins -- an element is positive \emph{if and
only if} it equals $g^\dagger g$ for some $g$ in the algebra~\cite{LiLu2023}.
Positivity and sums of squares then coincide, and the only thing a finite
truncation loses is the \emph{degree} of $g$; the failure of a low-order
relaxation is a degree-truncation effect and the hierarchy converges.  This is
the situation for the variational reduced-density-matrix method of quantum
chemistry, whose $P,Q,G$ conditions are exactly the dual of the degree-two
sum-of-squares problem~\cite{LiLu2023,NPOhierarchy}.  For the Weyl algebra no
such representation exists, positivity is not equivalent to being a sum of
squares, and the classical gap in the position block survives at every order.
The gap is nonetheless thin: Navascu\'es \emph{et al.}~\cite{NGAPP2013} prove
that every element of the Weyl algebra nonnegative in the Schr\"odinger
representation is a limit of sums of squares, and trace the apparent
convergence of bosonic relaxations to rounding error.  Density is a statement
about the limit, not about a fixed truncation, and it is at fixed truncation
that we work.
That gap is what we exploit, and it is why the phenomenon is confined to
unbounded systems: the analogous construction is unavailable for spins and
fermions, however poorly their low-order relaxations may perform.  Bosonic
matrix quantum mechanics is not in this class: its variables obey canonical
commutation relations at every rank, so the same trace argument applies and
the gap survives.

\subsection{Symmetry reduction, and why it is not a restriction}

Every potential considered in this paper is invariant under a finite group $G$
of $*$-automorphisms generated by the sign flips $q_j\to-q_j$ together with a
(possibly trivial) subgroup of coordinate permutations.  Two exact
consequences follow: moments vanish unless every $a_j+c_j$ is even, and with
$\pi(\alpha)=(a_j+c_j)\bmod 2$ the moment matrix is block diagonal in the
classes of $\pi$, with no change of basis.

The reduction is not merely computational.  The reduced feasible set is a
\emph{subset} of the full one, so a spurious point found in the symmetric
sector is spurious for the full bootstrap.  Conversely, restricting the search
for separating polynomials to $G$-invariant ones \emph{loses nothing}: the
group average of a nonnegative polynomial is nonnegative and has the same
expectation against an invariant moment vector.  Consequently a negative
search over the complete invariant family at a given degree establishes that
\emph{no} witness of that degree separates, not merely that none was found.

\section{The stationary bootstrap is not tight}
\label{sec:theorem}

Fix $d=2$, with rational potential and energy
\begin{equation}
  V = 20\big[(q_1^2-1)^2+(q_2^2-1)^2\big] - q_1^2q_2^2,
  \qquad E = 12 .
  \label{eq:system}
\end{equation}
The potential has four minima, at $(\pm q_*,\pm q_*)$ with $q_*^2=40/39$; its
ground-state energy is $E_0=11.22078$, obtained by Rayleigh--Ritz in a
harmonic-oscillator basis with the oscillator frequency matched to the well, so
the chosen $E$ lies above it.  The quartic part
$20(q_1^4+q_2^4)-q_1^2q_2^2$ is positive definite, so $H$ is bounded below
with discrete spectrum and the premise of Section~\ref{sec:search} is met.

We stress that $E$ need not be an eigenvalue: a feasible point that is not a
state is a failure of the relaxation at whatever energy it occurs.  Nor is the
failure vacuous for want of states at that energy.  A stationary mixed state
realizes any mean energy above $E_0$, so states satisfying (S), $\ev{1}=1$ and
$\ev{H}=12$ exactly do exist, and $\mathcal{F}_3^{\mathrm S}(12)$ is meant to
contain them.  It contains more: the spectrum of \eqref{eq:system} begins with
a near-degenerate quadruplet ending at $11.25041$ and has nothing further until
$21.57045$, so $E=12$ sits in a gap of width $10.3$, and level three does not
exclude it.

\begin{theorem}
There exists a moment vector $m$ with entries in $\Q$ such that
\begin{enumerate}
\item every level-three constraint of type \textup{(S)} for
      \eqref{eq:system} holds exactly, together with $\ev{1}=1$ and
      $\ev{H}=E$;
\item all four blocks of $M_3$ are positive definite;
\item $L_m(\pcert) = -\tfrac{10649385831371}{122880000000000}
      = -0.0866649238\ldots < 0$, where $\pcert$ is a polynomial with rational
      coefficients satisfying $\pcert\ge0$ on all of $\R^2$.
\end{enumerate}
Consequently $m$ is not the moment sequence of any positive measure, hence of
no quantum state, pure or mixed, yet it is feasible for the level-three
stationary bootstrap.
\end{theorem}

Every assertion is verified in exact rational arithmetic: the equalities by
substitution over $\Q$, positive definiteness by exact $LDL^\top$ (smallest
pivot $1.721\times10^{-3}$), and the nonnegativity of $\pcert$ by the
certificate constructed in Section~\ref{sec:witness}.  No floating-point
computation enters the statement.  The
moment vector itself, the Gram matrix, and a standalone script that re-checks
all three assertions are provided as ancillary files.  That script builds (S)
alone, deliberately: a reader who adds the (E) rows will find the point
infeasible, which reproduces Section~\ref{sec:eigenstate} rather than refuting
Theorem~1.

That the conclusion covers mixed states as well as pure ones is worth making
explicit, since the stationary relaxation is exactly the one used for thermal
states.  For any density matrix $\rho$ the positional moments
$\mathrm{tr}(\rho\,q^a)$ are the moments of a genuine positive measure, namely
the position-space diagonal of $\rho$.  Hence $L(\pcert)\ge0$ is necessary for
\emph{every} state, and its violation rules out all of them at once.

It is worth stating precisely what this does and does not cover.  Conditions
\eqref{eq:stat} and \eqref{eq:pos} are the core of the thermal bootstrap of
Ref.~\cite{CGSY2025}, whose semidefinite program imposes positivity of the
thermal density matrix, normalization, canonical relations, and the stationary
state conditions $\ev{[H,\mathcal{O}]}_\beta=0$ -- and no eigenstate
condition, since \eqref{eq:eig} is false for a thermal state.  Their warm-up
example is a one-dimensional anharmonic oscillator and their main application
is matrix quantum mechanics, so the formulation is not tied to either setting.
That formulation additionally imposes the Kubo--Martin--Schwinger condition, a
genuine further constraint, nonlinear in the expectation values and implemented
through a semidefinite relaxation of the matrix logarithm.  Theorem~1 therefore
establishes non-tightness of the stationarity-plus-positivity relaxation and
leaves open whether the KMS condition closes the gap.  We regard that as the
natural next question, and note that KMS carries explicit temperature
dependence, which none of the constraints considered here do.

\subsection{The separating polynomial}
\label{sec:witness}

For a given feasible point, the best separating polynomial within the
six-dimensional space of invariant sextics spanned by
$1$, $q_1^2+q_2^2$, $q_1^4+q_2^4$, $q_1^2q_2^2$, $q_1^6+q_2^6$ and
$q_1^4q_2^2+q_1^2q_2^4$ is the one minimizing $L_m(p)$ subject to $p\ge0$
pointwise, together with a norm bound on the coefficients, which is needed
only because the objective is homogeneous in them.  This is a linear program,
with pointwise nonnegativity imposed by cutting planes.  Alternating this with the SDP
converges to
\begin{align}
  p \;=\;& 7.88174 - 4.25817\,(x^2{+}y^2) + 0.57512\,(x^4{+}y^4) \notag\\
        &- 7.99889\,x^2y^2 + 3.74206\,(x^4y^2{+}x^2y^4),
  \label{eq:witness}
\end{align}
with an overall scale that is immaterial, since only the sign of $L_m(p)$ is at
issue; on the true ground state $L(p)=1.32>0$.  Its global
minimum is $-9.846\times10^{-5}$, attained on the axes.  Note that $p$ cannot be
a sum of squares: a sextic sum of squares uses squares of cubics only, and
$M_3$ is indexed by exactly those monomials, so were $p$ one, $L_m(p)\ge0$
would follow automatically from $M_3\succeq0$ (as can also be verified
directly by solving an SDP).

Two flaws must be repaired for an exact claim.  The polynomial dips slightly
negative, and its leading form $3.742\,x^2y^2(x^2+y^2)$ vanishes on both axes,
so its homogenization is not positive definite and Reznick's
theorem~\cite{Reznick1995} -- that $(\sum_j x_j^2)^N f$ is a sum of squares for
some $N$ whenever the form $f$ is positive definite -- does not apply.  Both are fixed at once by
\begin{equation}
  \pcert \;=\; p \;+\; \tfrac{1}{4096} \;+\; \tfrac{1}{65536}\,(x^6+y^6),
\end{equation}
available for free because the $x^6+y^6$ coefficient of $p$ vanishes.  Like
$p$, $\pcert$ is not itself a sum of squares, so a multiplier is needed.  We
then find that $N=1$ suffices: $(x^2+y^2+w^2)\,P_{\mathrm{c}}(x,y,w)$ is a sum
of squares, where $P_{\mathrm{c}}$ is the homogenization of $\pcert$.  Since the
multiplier is strictly positive away from the origin, this is what establishes
$\pcert\ge0$ on $\R^2$.  The $15\times15$ Gram matrix is rounded to rationals, projected exactly onto the coefficient-matching
subspace, and verified positive definite by exact $LDL^\top$,
following~\cite{PeyrlParrilo}.

Raising the truncation repairs this instance.  At $K=4$, with the same
potential and energy and the witness optimized over the complete
nine-dimensional space of invariant octics, the alternating search over
$\mathcal{F}_4^{\mathrm S}(12)$ returns no separating witness, at two
coordinate scalings.  That is how a hierarchy is expected to behave, and it
should not be mistaken for the content of Theorem~1.  What fails at level
three is not the degree of the certificate but the cone: were $\pcert$ a sum
of squares it would be a sum of squares of cubics, and $M_3$ is indexed by
exactly those monomials, so $L_m(\pcert)\ge0$ would already follow from
$M_3\succeq0$.  Since $M_3$ carries every monomial in $q$ \emph{and} $p$ of
degree at most three, the same argument shows more: $\pcert$ is not a sum of
hermitian squares of degree three in the Weyl algebra either, so the
approximation of~\cite{NGAPP2013} cannot be reached at this level.  The relaxation carries every monomial the certificate would
need and still admits a point that is the moment sequence of no state.  The
gap exploited here is between the cones, not between the degrees.

\section{The eigenstate constraints appear to close the gap}
\label{sec:eigenstate}

Adding the (E) rows removes the violation.  At \eqref{eq:system} the rational
point of Theorem~1 gives $L(\pcert)=-0.0867$, whereas the \emph{minimum} of
$L(\pcert)$ over the eigenstate feasible set is $+1.06$ -- stable across three
coordinate scalings and two positivity margins.

Very little is needed to produce that reversal.  At $K=3$ the degree budget
in \eqref{eq:constraints} admits fifteen monomials $\mathcal{O}$ into (E),
and exact row reduction over $\Q$ shows that the resulting rows raise the rank
of the constraint system by only three, from $52$ to $55$.  Two of those three
new directions -- those carried by $\mathcal{O}=q_j^2$ and
$\mathcal{O}=p_j^2$ -- are precisely the ones the spurious point of Theorem~1
violates; it satisfies the third, from $\mathcal{O}=q_jp_j$, exactly.  Two
independent linear conditions therefore separate a verified spurious point
from a feasible set on which the witness is bounded away from zero.

More tellingly, the spurious moment vector of Theorem~1 violates the (E)
constraints by a wide margin.  Taking $\mathcal{O}=q_j^2$ and
$\mathcal{O}=p_j^2$ in turn, exact arithmetic gives
\begin{align}
  \ev{q_j^2H}-E\ev{q_j^2} &= \tfrac{48650999}{600000} = +81.08,\notag\\
  \ev{p_j^2H}-E\ev{p_j^2} &= \tfrac{393646698559}{2400000} = +164019.46,
\end{align}
for $j=1,2$; the remaining monomials of degree at most two are satisfied.  The reason is visible in the moments
themselves: $\ev{q_1^2}=1.015$, $\ev{q_1^4}=1.144$, $\ev{q_1^6}=3.97$ are
of ordinary size, and $\ev{p_1^2}=8.40$ only modestly large, but
$\ev{p_1^4}=2.0\times10^5$ and $\ev{p_1^6}=9.4\times10^9$ exceed their
ground-state values by three and six orders of magnitude.

\subsection{Systematic search}
\label{sec:search}

We searched for a violation of the eigenstate bootstrap in the settings of
Table~\ref{tab:settings}.  In the first three, the separating polynomial is
optimized over the complete space of invariant forms of the maximal degree the
truncation can see, so a null result means no witness of that degree exists;
in the $d=5$ settings the family is the invariant quartics, where the
copositive analysis of Section~\ref{sec:copositive} applies.  Five settings
give a null result.  The sixth, at $d=5$ and $K=2$, does produce a violation;
but there the only surviving constraint of type (E) is $\ev{H}=E$ itself, since
$\deg\mathcal{O}\le 2K-\deg V=0$, so the eigenstate conditions are almost
entirely absent.  Section~\ref{sec:copositive} shows why the copositive
picture leaves room for a violation at $K=2$, and why it disappears at $K=3$
for a structural reason.

\begin{table}[t]
\begin{tabular}{llll}
\toprule
$d$ & $K$ & witness family & result \\
\midrule
2 & 3 & invariant sextics (complete)   & none \\
2 & 4 & invariant octics (complete), 5 potentials & none \\
3 & 3 & invariant sextics (complete), 3 potentials & none \\
5 & 2 & invariant quartics, Horn        & \emph{$-0.968$}$^\ast$ \\
5 & 3 & invariant quartics, Horn        & none \\
5 & 3 & scaled Horn, 58 bounded potentials & none \\
\bottomrule
\end{tabular}
\caption{Search for a violation of the eigenstate bootstrap; each entry is
$\min L$ over the feasible set, and ``none'' records a nonnegative minimum,
i.e.\ that no witness in the family separates.  Five settings give no
violation.  The witness family is complete at the relevant degree in the first
three settings; at $d=5$ it consists of invariant quartics, for the reason
given in Section~\ref{sec:copositive}.
$^\ast$At $K=2$ with quartic $V$ the degree budget leaves $\ev{H}=E$ as the
only surviving constraint of type (E), so this row tests a relaxation with
essentially no eigenstate content; see Sections~\ref{sec:search} and
\ref{sec:copositive}.}
\label{tab:settings}
\end{table}

Two entries deserve comment.  The $d=3$ family contains Robinson's
form~\cite{Robinson}, which is nonnegative and not SOS, and the potential was
chosen with minima on its zero set, so the obstruction is genuinely available;
this follows the strategy of~\cite{AKR2021}, where an instance is engineered so
that the physical configuration lies on the zero set of a nonnegative non-SOS
form.
The $d=5$ entries are discussed in the next section.

The $d=5$, $K=3$ sweep is the most extensive: $64$ asymmetric ring potentials
with pre-screened witnesses, with on-site coefficient $a$ and nearest-neighbour
coupling $c$ as in Section~\ref{sec:copositive}.  Six returned a negative value.  All six lie
exactly on the boundary $a+c=0$ of the region where the quartic part of $V$ is
positive definite; the corresponding Hamiltonians are unbounded below and have
no discrete spectrum, so the bootstrap premise fails and a feasible point
proves nothing.  All six were also reported as \texttt{optimal\_inaccurate} by
the solver, while all $58$ bounded instances returned \texttt{optimal} with
$L\ge+0.349$.  The physical criterion and the solver's own status flag thus
select exactly the same six instances.

The practical conclusion: \emph{before accepting any bootstrap violation,
verify that $H$ is bounded below with discrete spectrum}.  A
relaxation is trivially non-tight for a Hamiltonian with no eigenstates, and
nothing in the solver output distinguishes that case.  We note that unbounded
potentials are not merely a pitfall: Ref.~\cite{CGSY2025} studies metastable
thermal states of a matrix model whose potential is unbounded below, and uses
SDP infeasibility above a critical temperature to bound the point at which the
thermal state ceases to exist.  The distinction that matters is whether the
absence of a state is the conclusion being drawn or an unnoticed premise.

\section{The mechanism: bounded momentum moments}
\label{sec:mechanism}

The pattern above has a cause, and it can be measured.  On the level-$K$
feasible set of the two-dimensional double well we computed the supremum of
$\ev{p_1^{2k}}$ with and without the (E) rows; Table~\ref{tab:pbound} collects
the results.

\begin{table}[t]
\begin{tabular}{llrrr}
\toprule
$E$ & rows & $\sup\ev{p_1^2}$ & $\sup\ev{p_1^4}$ & $\sup\ev{p_1^6}$\\
\midrule
12 & (S) only & 10.98 & unbounded & unbounded\\
12 & (S)+(E)  &  7.70 & 196.3     & unbounded\\
16 & (S) only & 14.66 & unbounded & unbounded\\
16 & (S)+(E)  & 11.76 & 328.6     & unbounded\\
\bottomrule
\end{tabular}
\caption{Level $K=3$, two-dimensional quartic double well.  For comparison, the
true ground state of \eqref{eq:system} has $\ev{p_1^2}=5.75$,
$\ev{p_1^4}=1.08\times10^2$ and $\ev{p_1^6}=3.42\times10^3$, while the spurious
point of Theorem~1 has $\ev{p_1^2}=8.40$, $\ev{p_1^4}=2.0\times10^5$ and
$\ev{p_1^6}=9.4\times10^9$.}
\label{tab:pbound}
\end{table}

The fourth momentum moment is \emph{unbounded} without the eigenstate
constraints and finite with them, at the right order of magnitude.  An
independent formulation -- testing feasibility of $\ev{p_1^4}\ge T$ on a
ladder of thresholds, which avoids the unboundedness that defeats a direct
maximization -- agrees: at $E=16$ the stationary set admits
$\ev{p_1^4}\ge10^8$, while the eigenstate set is already infeasible at
$\ev{p_1^4}\ge10^3$.  At least five orders of magnitude separate the two.

The spurious point of Theorem~1 sits far out along precisely this unbounded
direction, and nothing in the stationary relaxation registers it.  The position
block sees only $\ev{q^a}$ and cannot detect an inflated
$\ev{p^4}$; the commutator relations relate the two sectors but do not bound
either; only \eqref{eq:eig} does.  Note that $\ev{p^6}$ remains unbounded at
$K=3$, consistent with the degree budget, which reaches only
$\deg\mathcal{O}\le2$ there.  The budget widens with $K$: at $K=4$ the
(E) rows raise the rank of the constraint system from $119$ to $130$, against
$52$ to $55$ at $K=3$.

\section{Copositivity, and where a gap can open}
\label{sec:copositive}

The searches of Section~\ref{sec:search} failed in a structured way, and the
structure is informative.

Under the sign symmetry $q_j\to-q_j$, the trivial-parity position block at
$K=2$ has rows indexed by $\{1,q_1^2,\dots,q_d^2\}$ with entries
$\ev{q_i^2q_j^2}$, and every sign-invariant quartic witness takes the form
$\sum_{ij}W_{ij}x_i^2x_j^2$ for a symmetric matrix $W$.  Since the $x_i^2$
sweep out the nonnegative orthant, such a witness is nonnegative on $\R^d$
precisely when $W$ lies in the copositive cone
$\COP_d=\{W : x^\top\! Wx\ge0 \ \forall x\ge0\}$, whereas $M_K\succeq0$ certifies
only membership in $\PSDN$, the sum of the cone $\mathcal{P}$ of positive
semidefinite matrices and the cone $\mathcal{N}$ of the entrywise nonnegative
ones.  By Diananda's theorem~\cite{Diananda} these coincide for $d\le4$ and differ first
at $d=5$, where Horn's matrix~\cite{Diananda} is copositive but not in $\PSDN$.
This accounts for the null results at two and three degrees of freedom under
sign symmetry, and identifies $d=5$ as the first place to look.

The level of the bootstrap enters through Parrilo's
hierarchy~\cite{Parrilo2000}: $W\in \Kc^r$ iff
$(\sum_j x_j^2)^r\sum_{ij}W_{ij}x_i^2x_j^2$ is SOS, with $\Kc^0=\PSDN$ and
$\bigcup_r \Kc^r$ dense in $\COP_d$.  Level $K$ of the bootstrap certifies
membership in $\Kc^{K-2}$.  Consequently:

\begin{itemize}
\item at $K=2$, the witness need only lie outside $\Kc^0$;
\item at $K=3$, it must lie outside $\Kc^1$.
\end{itemize}

This accounts for all three of the $d=5$ entries in Table~\ref{tab:settings}.
Horn's matrix lies outside $\Kc^0$, so nothing prevents a violation at $K=2$, and
we find one, at
energies above the ground state, in a five-site ring
$V=a\sum_j(q_j^2-1)^2+c\sum_j q_j^2q_{j+1}^2$ with the dihedral symmetry that
Horn's circulant structure requires.  (Full permutation symmetry would be
fatal: the $S_5$ average of Horn's matrix is the identity, which is positive
semidefinite.)  But Horn's matrix lies \emph{inside} $\Kc^1$, and by a scaling
theorem~\cite{DDGH2013} every unit-diagonal copositive $5\times5$ matrix does.  A dihedrally invariant witness necessarily has constant, hence
unit-scalable, diagonal.  So $K=3$ at $d=5$ is closed to symmetric witnesses
for a structural reason, and indeed we find $L=+1.55$.

The same theorem says $\Kc^1$ is not invariant under diagonal scaling, and that
some scaling $DWD$ of any copositive matrix outside $\PSDN$ escapes $\Kc^1$.  We
verify this numerically and use such scalings as witnesses in an asymmetric
ring; this is the $58$-potential sweep, which found nothing.

At $d=6$ the corresponding statement fails: Hildebrand and
Afonin~\cite{HA2024} exhibit an extreme copositive $6\times6$ matrix with unit
diagonal lying outside $\Kc^1$.  Their example is not circulant, and a
dihedrally invariant witness is.  We generated approximately $1.3\times10^4$
random circulant $6\times6$ candidates, of which $5196$ passed a copositivity
test by projected gradient descent on the simplex; $\Kc^1$ membership was then
decided for each of these by a $56\times56$ SDP, and none lay outside $\Kc^1$.
The $d=6$ obstruction thus appears to be inaccessible to a symmetric ring, and
reaching it would require an asymmetric potential, as at $d=5$.

\section{Discussion}

We draw three conclusions.

First, the two relaxations must be distinguished.  The stationary bootstrap is
not tight at level three in two dimensions, and this is proved, exactly.  Since
that is the relaxation appropriate to thermal and mixed states, low-order
feasibility there is weaker evidence than it appears, and we recommend the
diagnostic below.

Second, the eigenstate bootstrap resisted every attack we could design, across
dimensions, truncation levels, and obstruction families, and always against
witness families that were either complete at the relevant degree or proved to
lie outside the cone the moment matrix can certify.  We emphasize that this is
evidence, not proof: the alternating search is a local method for a bilinear
problem.  What is established at each visited point is that the linear program
over the invariant family attains its optimum at zero, i.e.\ that no invariant
witness separates \emph{there}.

Third, and most usefully, there is a mechanism, and it suggests a theorem.
Table~\ref{tab:pbound} shows the eigenstate constraints converting an unbounded
momentum direction into a bounded one.  A bound of the form
$\ev{p^{2k}}\le C(E,k)$ is an Archimedean-type statement -- a boundedness
condition of the kind that convergence proofs require -- and its absence in the
Weyl algebra is exactly what places this problem outside the existing
theory~\cite{NPO2026,NPOhierarchy}.
We therefore conjecture:

\begin{conjecture}
Let $H$ be as in \eqref{eq:H}, bounded below with discrete spectrum.  There are
constants $C(E,k)$ such that every $m\in\mathcal{F}_K^{\mathrm{SE}}(E)$
satisfies $\ev{p_j^{2k}}\le C(E,k)$ for $k\le K-\tfrac12\deg V+1$, the case $k=2$, $K=3$,
$d=2$ with quartic $V$ being Proposition~1 below; and for all sufficiently
large $K$ the positional moments of every
$m\in\mathcal{F}_K^{\mathrm{SE}}(E)$ lie in the moment cone.
\end{conjecture}

The restriction on $k$ is a prediction of the degree budget: the operator
$\mathcal{O}=p_j^{2k-2}$ that produces the bound is admissible in (E) only when
$2k-2\le2K-\deg V$.  Table~\ref{tab:pbound} is consistent with this on both
sides: for the quartic $V$ of \eqref{eq:system} at $K=3$ the range is $k\le2$,
the bound appears for $k=2$, and $\ev{p^6}$ remains unbounded.

The dependence on $\deg V$ is not cosmetic.  When $\deg V=2K$ the budget
vanishes, (E) retains only $\ev{H}=E$, and no bound of this form can survive:
for $H=\tfrac12p^2+q^{2r}$ at $K=r$, a mixture of the ground state with a
single highly excited eigenstate, weighted so that the energy is $E$, commutes
with $H$ and is therefore feasible at every excitation, while the virial
relation for a homogeneous potential forces $\ev{p^{2k}}\to\infty$ along that
family for every $k\ge2$.  We thank Ye Zhou for this example.  It carries a
practical corollary: at $\deg V=2K$ the eigenstate constraints add nothing to
the energy, and the eigenstate bootstrap collapses to the stationary one.

At fixed $K$ the first half involves only finitely many moments, and for the
system of Theorem~1 it is a theorem.

\begin{proposition}
Let $d=2$ and let $V$ be quartic with $\min V=V_*$, and write
$V-V_*=\sum_k g_k^2$ for any sum-of-squares decomposition.  Then every
$m\in\mathcal{F}_3^{\mathrm{SE}}(E)$ satisfies $\ev{p_j^2}\le 2(E-V_*)$ and
\begin{equation}
  \ev{p_j^4}\ \le\ 2(E-V_*)\,\ev{p_j^2}
    \;+\; \sum_k\ev{\partial_j^2(g_k^2)} .
  \label{eq:pbound}
\end{equation}
\end{proposition}

\begin{proof}
Take $\mathcal{O}=p_j^2$, admissible since $\deg\mathcal{O}\le2K-\deg V=2$.
Then \eqref{eq:eig} reads
$\tfrac12\sum_i\ev{p_j^2p_i^2}+\mathrm{Re}\,\ev{p_j^2V}=E\ev{p_j^2}$.  The
momenta commute, so $p_j^2p_i^2=(p_ip_j)^\dagger(p_ip_j)$ is a diagonal entry
of $M_3$ and is nonnegative; discarding all but $i=j$ gives
$\tfrac12\ev{p_j^4}\le E\ev{p_j^2}-\mathrm{Re}\,\ev{p_j^2V}$.  For any real
polynomial $g$,
\begin{equation}
  \tfrac12\{p_j^2,g^2\} \;=\; (gp_j)^\dagger(gp_j)
     \;-\; \tfrac12\,\partial_j^2(g^2),
  \label{eq:gpid}
\end{equation}
so $\mathrm{Re}\,\ev{p_j^2g_k^2}\ge-\tfrac12\ev{\partial_j^2(g_k^2)}$, since
$g_kp_j$ has degree three and its square therefore has nonnegative
expectation.  Summing over $k$ and using $V=V_*+\sum_kg_k^2$ gives
\eqref{eq:pbound}; the bound on $\ev{p_j^2}$ is $\ev{H}=E$ with $V\ge V_*$.
\end{proof}

The first bound uses only $\ev{H}=E$ and $M_3\succeq0$, so it holds on
$\mathcal{F}_3^{\mathrm{S}}(E)$ as well; that is why the $\ev{p_1^2}$ column of
Table~\ref{tab:pbound} is finite in all four rows, and it places the mechanism
at $\ev{p^4}$ rather than $\ev{p^2}$.

Two features of the argument are worth isolating.  Every operator it uses is
carried at $K=3$: $p_ip_j$ and $g_kp_j$ have degrees two and three, inside the
moment matrix, and $\partial_j^2(g_k^2)$ has degree two.  And the single
hypothesis is that $V-V_*$ be a sum of squares, which for a quartic in two
variables is automatic by Hilbert's classification and fails in general
precisely in the regime this paper is about.  The obstruction to proving
Conjecture~1 as stated is thus the same gap the conjecture would close.  For
\eqref{eq:system} one may take
$V-V_*=\tfrac{1600}{39}\bigl[1-\tfrac{39}{80}(q_1^2+q_2^2)\bigr]^2
 +\tfrac{41}{4}\,(q_1^2-q_2^2)^2$ with $V_*=-40/39$, whose Gram matrix in the
basis $(1,q_1^2,q_2^2)$ is positive semidefinite of rank two -- its
determinant vanishes, as it must, since $V-V_*$ has four zeros -- and
$\sum_k\partial_1^2(g_k^2)=240q_1^2-2q_2^2-80$.  The right-hand side of
\eqref{eq:pbound} is itself bounded on the feasible set, so the bound is a
constant: $\ev{H}=E$ with $\ev{p_i^2}\ge0$ gives $\ev{V}\le E$, and
$V-\tfrac{39}{4}(q_1^4+q_2^4)+43$ is a nonnegative quartic in two variables,
hence a sum of squares of quadratics and of nonnegative expectation under
$M_3\succeq0$, so that $\ev{q_1^4}\le5.65$ and, by the $2\times2$ minor of
$M_3$ on $(1,q_1^2)$, $\ev{q_1^2}\le2.38$.  At $E=12$ this gives
$\ev{p_1^4}\le1.2\times10^3$, against the supremum $196.3$ of
Table~\ref{tab:pbound}: explicit, and not tight.

A strict Positivstellensatz does exist for the Weyl
algebra~\cite{Schmudgen2005}, but its hypotheses (strict positivity and an
elliptic leading symbol) exclude precisely the boundary cases at issue.

We also note where we would look next.  A construction of this kind needs a
genuine gap between positivity and sums of squares together with an objective
assembled from tunable data.  By the argument of Section~\ref{sec:momentum} the
first requirement excludes every setting with finite-dimensional representations, which
rules out spins, fermionic models, and the variational reduced density matrix
method.  Multi-matrix quantum mechanics has both the gap and the tunable
data.  Its single-trace moment
matrix is a truncated \emph{tracial} moment matrix, so positivity there
certifies only that a polynomial is cyclically equivalent to a sum of hermitian
squares, whereas a physical state need only be trace-positive, that is, have nonnegative
trace on every tuple of symmetric matrices.  These cones
already differ in two noncommuting variables~\cite{KS2008}, and the threshold
is sharp: in two variables and degree at most four they coincide~\cite{BK2010},
in analogy with Hilbert's classification.  Whether a matrix bootstrap at the
truncations in use admits such a witness is the natural next question, and it
is again the stationary relaxation that should be tested first.

\emph{A diagnostic.}  Given a feasible point $m$, minimizing $L_m(f)$ over a
fixed family of nonnegative non-SOS polynomials is a linear program, and a
negative value certifies that $m$ is the moment sequence of no state.  The
exposure is uneven.  A one-sided bound obtained by optimizing over the feasible
set -- as in the stationary bosonic relaxations of~\cite{RobichonTilloy2024},
which are multidimensional and of low order -- remains valid however large that
set is; what a cone gap costs there is convergence, not correctness.  It is
feasibility and exclusion claims that fail outright, and we make no claim
against any published result;
but the test is cheap, and we suggest running it whenever such a result is
reported in more
than one dimension -- and, per
Section~\ref{sec:search}, checking first that the Hamiltonian is bounded
below.

\begin{acknowledgments}
We thank Ye Zhou for pointing out that the degree range stated in Conjecture~1
is specific to quartic potentials, and for supplying the counterexample that
fixes it.
\end{acknowledgments}

\appendix

\section{Computational details}

The Weyl algebra is implemented in normal order using
$p^cq^a=\sum_k k!\binom{a}{k}\binom{c}{k}(-i)^k q^{a-k}p^{c-k}$; products and
adjoints agree with matrix representations to $10^{-13}$, once the
representations are formed in a basis padded by the operator degree so that
truncation does not interfere.  For the exact stage all structure constants are
Gaussian rationals, represented as pairs of exact fractions.

The exact constraint system at $K=3$, $d=2$ has $218$ rows in $74$ unknowns
with rank $52$.  Exact row reduction yields a particular solution $m_0$ and a
null-space basis $N$, so that $m=m_0+\sum_j u_jN_j$ satisfies the constraints
exactly for \emph{any} rational $u$; the floating-point optimum is then
expressed in these coordinates and rounded.

Several pitfalls deserve mention, as each cost us a result.  Solving in
unscaled coordinates returns a value well above the true minimum while
reporting successful convergence.  A solver once reported optimality at a point
violating positive semidefiniteness by $5\times10^{-2}$.  The moment matrix of a
true ground state can have $\lambda_{\min}\sim10^{-6}$, so any positivity
margin above that excludes the physical state and renders the SDP spuriously
infeasible.  And every build should be validated against an exactly known
state: evaluating the constraint system on the moments of a harmonic-oscillator
ground state, which factorize across modes and are therefore exact in any
dimension, reproduces the constraint system $Cz=b$ to $10^{-16}$, catching build errors
that no amount of reading the code will reveal.

Semidefinite programs were modelled with CVXPY~\cite{cvxpy} and solved with
Clarabel~\cite{clarabel}; the moment-matrix relaxations follow the Lasserre
hierarchy~\cite{Lasserre2001}.  The explicit rational moment vector, the
rational Gram matrix certifying the witness, and a standalone script that
re-verifies Theorem~1 from scratch using only the Python standard library are
included as ancillary files.

\bibliographystyle{apsrev4-2}
\bibliography{tightness}

@article{HHK2020,
  author  = {Han, Xizhi and Hartnoll, Sean A. and Kruthoff, Jorrit},
  title   = {Bootstrapping Matrix Quantum Mechanics},
  journal = {Phys. Rev. Lett.},
  volume  = {125},
  pages   = {041601},
  year    = {2020},
  eprint  = {2004.10212},
  archivePrefix = {arXiv},
  primaryClass  = {hep-th}
}

@unpublished{BH2021,
  author  = {Berenstein, David and Hulsey, George},
  title   = {Bootstrapping simple quantum mechanical systems},
  year    = {2021},
  eprint  = {2108.08757},
  archivePrefix = {arXiv},
  primaryClass  = {hep-th}
}

@article{BH2022a,
  author  = {Berenstein, David and Hulsey, George},
  title   = {Bootstrapping more {QM} systems},
  journal = {J. Phys. A},
  volume  = {55},
  number  = {27},
  pages   = {275304},
  year    = {2022},
  eprint  = {2109.06251},
  archivePrefix = {arXiv},
  primaryClass  = {hep-th}
}

@article{BH2023,
  author  = {Berenstein, David and Hulsey, George},
  title   = {Semidefinite programming algorithm for the quantum mechanical
             bootstrap},
  journal = {Phys. Rev. E},
  volume  = {107},
  pages   = {L053301},
  year    = {2023},
  eprint  = {2209.14332},
  archivePrefix = {arXiv},
  primaryClass  = {hep-th}
}

@article{BHreflection,
  author  = {Berenstein, David and Hulsey, George},
  title   = {One-dimensional reflection in the quantum mechanical bootstrap},
  journal = {Phys. Rev. D},
  volume  = {109},
  pages   = {025013},
  year    = {2024},
  eprint  = {2307.11724},
  archivePrefix = {arXiv},
  primaryClass  = {hep-th}
}

@article{BHhalfline,
  author  = {Berenstein, David and Hulsey, George},
  title   = {Anomalous bootstrap on the half-line},
  journal = {Phys. Rev. D},
  volume  = {106},
  pages   = {045029},
  year    = {2022},
  eprint  = {2206.01765},
  archivePrefix = {arXiv},
  primaryClass  = {hep-th}
}

@article{CGSY2025,
  author  = {Cho, Minjae and Gabai, Barak and Sandor, Joshua and Yin, Xi},
  title   = {Thermal bootstrap of matrix quantum mechanics},
  journal = {JHEP},
  volume  = {2025},
  number  = {4},
  pages   = {186},
  year    = {2025},
  doi     = {10.1007/JHEP04(2025)186},
  eprint  = {2410.04262},
  archivePrefix = {arXiv},
  primaryClass  = {hep-th}
}

@article{spurious,
  author  = {Bhattacharya, Jyotirmoy and Das, Diptarka and Das, Sayan Kumar
             and Jha, Anuj Kumar and Kundu, Moulindu},
  title   = {Numerical bootstrap in quantum mechanics},
  journal = {Phys. Lett. B},
  volume  = {823},
  pages   = {136785},
  year    = {2021},
  eprint  = {2108.11416},
  archivePrefix = {arXiv},
  primaryClass  = {hep-th}
}

@unpublished{central2025,
  author  = {Lawrence, Scott and McPeak, Brian},
  title   = {Quantum bootstrap for central potentials},
  year    = {2025},
  eprint  = {2512.09041},
  archivePrefix = {arXiv},
  primaryClass  = {quant-ph}
}

@article{Hilbert1888,
  author  = {Hilbert, David},
  title   = {{\"U}ber die Darstellung definiter Formen als Summe von
             Formenquadraten},
  journal = {Math. Ann.},
  volume  = {32},
  pages   = {342--350},
  year    = {1888}
}

@incollection{Motzkin,
  author    = {Motzkin, Theodore S.},
  title     = {The arithmetic-geometric inequality},
  booktitle = {Inequalities (Proc. Sympos. Wright-Patterson AFB, 1965)},
  publisher = {Academic Press},
  address   = {New York},
  pages     = {205--224},
  year      = {1967}
}

@article{ChoiLam,
  author  = {Choi, Man-Duen and Lam, Tsit-Yuen},
  title   = {Extremal positive semidefinite forms},
  journal = {Math. Ann.},
  volume  = {231},
  pages   = {1--18},
  year    = {1977}
}

@incollection{Robinson,
  author    = {Robinson, Raphael M.},
  title     = {Some definite polynomials which are not sums of squares of
               real polynomials},
  booktitle = {Selected Questions of Algebra and Logic},
  publisher = {Izdat. Nauka, Novosibirsk},
  pages     = {264--282},
  year      = {1973}
}

@article{Blekherman,
  author  = {Blekherman, Grigoriy},
  title   = {There are significantly more nonnegative polynomials than sums
             of squares},
  journal = {Israel J. Math.},
  volume  = {153},
  pages   = {355--380},
  year    = {2006}
}

@article{Diananda,
  author  = {Diananda, Palahenedi Hewage},
  title   = {On non-negative forms in real variables some or all of which are
             non-negative},
  journal = {Proc. Cambridge Philos. Soc.},
  volume  = {58},
  pages   = {17--25},
  year    = {1962}
}

@phdthesis{Parrilo2000,
  author  = {Parrilo, Pablo A.},
  title   = {Structured Semidefinite Programs and Semialgebraic Geometry
             Methods in Robustness and Optimization},
  school  = {California Institute of Technology},
  year    = {2000},
  doi     = {10.7907/2K6Y-CH43}
}

@article{DDGH2013,
  author  = {Dickinson, Peter J. C. and D{\"u}r, Mirjam and Gijben, Luuk
             and Hildebrand, Roland},
  title   = {Scaling relationship between the copositive cone and Parrilo's
             first level approximation},
  journal = {Optim. Lett.},
  volume  = {7},
  pages   = {1669--1679},
  year    = {2013}
}

@article{Helton2002,
  author  = {Helton, J. William},
  title   = {``Positive'' noncommutative polynomials are sums of squares},
  journal = {Ann. of Math. (2)},
  volume  = {156},
  number  = {2},
  pages   = {675--694},
  year    = {2002}
}

@book{BKP2016,
  author    = {Burgdorf, Sabine and Klep, Igor and Povh, Janez},
  title     = {Optimization of Polynomials in Non-Commuting Variables},
  series    = {SpringerBriefs in Mathematics},
  publisher = {Springer},
  year      = {2016}
}

@article{NPO2026,
  author  = {Ara{\'u}jo, Mateus and Klep, Igor and Garner, Andrew J. P.
             and V{\'e}rtesi, Tam{\'a}s and Navascu{\'e}s, Miguel},
  title   = {First-order optimality conditions for non-commutative
             optimization problems},
  journal = {Found. Comput. Math.},
  year    = {2026},
  doi     = {10.1007/s10208-026-09761-x}
}

@article{NPOhierarchy,
  author  = {Pironio, Stefano and Navascu{\'e}s, Miguel and Ac{\'i}n, Antonio},
  title   = {Convergent relaxations of polynomial optimization problems with
             noncommuting variables},
  journal = {SIAM J. Optim.},
  volume  = {20},
  pages   = {2157--2180},
  year    = {2010}
}

@article{Reznick1995,
  author  = {Reznick, Bruce},
  title   = {Uniform denominators in {H}ilbert's seventeenth problem},
  journal = {Math. Z.},
  volume  = {220},
  pages   = {75--97},
  year    = {1995}
}

@article{PeyrlParrilo,
  author  = {Peyrl, Helfried and Parrilo, Pablo A.},
  title   = {Computing sum of squares decompositions with rational
             coefficients},
  journal = {Theor. Comput. Sci.},
  volume  = {409},
  pages   = {269--281},
  year    = {2008}
}

@article{Schmudgen2005,
  author  = {Schm{\"u}dgen, Konrad},
  title   = {A strict {P}ositivstellensatz for the {W}eyl algebra},
  journal = {Math. Ann.},
  volume  = {331},
  pages   = {779--794},
  year    = {2005},
  eprint  = {math/0403076},
  archivePrefix = {arXiv},
  primaryClass  = {math.AG}
}

@article{AKR2021,
  author  = {Alfassi, Yuval and Keren, Daniel and Reznick, Bruce},
  title   = {The Non-Tightness of a Convex Relaxation to Rotation Recovery},
  journal = {Sensors},
  volume  = {21},
  pages   = {7358},
  year    = {2021}
}

@article{cvxpy,
  author  = {Diamond, Steven and Boyd, Stephen},
  title   = {{CVXPY}: A {P}ython-embedded modeling language for convex
             optimization},
  journal = {J. Mach. Learn. Res.},
  volume  = {17},
  pages   = {1--5},
  year    = {2016}
}

@unpublished{clarabel,
  author  = {Goulart, Paul J. and Chen, Yuwen},
  title   = {Clarabel: An interior-point solver for conic programs with
             quadratic objectives},
  year    = {2024},
  eprint  = {2405.12762},
  archivePrefix = {arXiv},
  primaryClass  = {math.OC}
}

@article{Lasserre2001,
  author  = {Lasserre, Jean B.},
  title   = {Global optimization with polynomials and the problem of moments},
  journal = {SIAM J. Optim.},
  volume  = {11},
  pages   = {796--817},
  year    = {2001}
}

@unpublished{LiLu2023,
  author  = {Li, Bowen and Lu, Jianfeng},
  title   = {Quantum variational embedding for ground-state energy problems:
             sum of squares and cluster selection},
  year    = {2023},
  eprint  = {2305.18571},
  archivePrefix = {arXiv},
  primaryClass  = {quant-ph}
}

@article{KS2008,
  author  = {Klep, Igor and Schweighofer, Markus},
  title   = {Sums of Hermitian squares and the {BMV} conjecture},
  journal = {J. Stat. Phys.},
  volume  = {133},
  number  = {4},
  pages   = {739--760},
  year    = {2008},
  doi     = {10.1007/s10955-008-9632-x},
  eprint  = {0710.1074},
  archivePrefix = {arXiv},
  primaryClass  = {math.OA}
}

@article{BK2010,
  author  = {Burgdorf, Sabine and Klep, Igor},
  title   = {Trace-positive polynomials and the quartic tracial moment problem},
  journal = {C. R. Math. Acad. Sci. Paris},
  volume  = {348},
  number  = {13--14},
  pages   = {721--726},
  year    = {2010}
}

@article{HA2024,
  author  = {Hildebrand, Roland and Afonin, Andrey},
  title   = {On the structure of the $6\times 6$ copositive cone},
  journal = {Linear Algebra Appl.},
  volume  = {693},
  pages   = {22--38},
  year    = {2024},
  doi     = {10.1016/j.laa.2023.02.004},
  eprint  = {2209.08039},
  archivePrefix = {arXiv},
  primaryClass  = {math.OC}
}

@article{NGAPP2013,
  author  = {Navascu\'es, Miguel and Garc\'ia-S\'aez, Artur and Ac\'in, Antonio
             and Pironio, Stefano and Plenio, Martin B.},
  title   = {A paradox in bosonic energy computations via semidefinite
             programming relaxations},
  journal = {New J. Phys.},
  volume  = {15},
  pages   = {023026},
  year    = {2013},
  doi     = {10.1088/1367-2630/15/2/023026},
  eprint  = {1203.3777},
  archivePrefix = {arXiv},
  primaryClass  = {quant-ph}
}

@article{RobichonTilloy2024,
  author  = {Robichon, Gustave and Tilloy, Antoine},
  title   = {Bootstrapping the stationary state of bosonic open quantum
             systems},
  journal = {Phys. Rev. A},
  year    = {2026},
  doi     = {10.1103/5db2-q49l},
  eprint  = {2410.07384},
  archivePrefix = {arXiv},
  primaryClass  = {quant-ph}
}

@unpublished{Morita2026,
  author  = {Morita, Takeshi and Piensuk, Worapat and Soni, Pushkar},
  title   = {Ambiguity problem of the bootstrap method in quantum mechanics},
  year    = {2026},
  eprint  = {2605.30536},
  archivePrefix = {arXiv},
  primaryClass  = {hep-th}
}

@article{NancarrowXin2023,
  author  = {Nancarrow, Colin Oscar and Xin, Yuan},
  title   = {Bootstrapping the gap in quantum spin systems},
  journal = {JHEP},
  volume  = {2023},
  number  = {8},
  pages   = {052},
  year    = {2023},
  doi     = {10.1007/JHEP08(2023)052},
  eprint  = {2211.03819},
  archivePrefix = {arXiv},
  primaryClass  = {hep-th}
}

\end{document}